\documentclass[12pt,draftcls,onecolumn]{IEEEtran}
\usepackage{latexsym}
\usepackage{graphicx}
\usepackage{array}
\usepackage{amsmath}
\usepackage{amsfonts}
\usepackage{amssymb}
\usepackage{amsthm}
\usepackage{flushend}

\newtheorem{thm}{Theorem}
\newtheorem{lemma}{Lemma}
\newtheorem{prop}{Proposition}
\newtheorem{defn}{Definition}

\newtheorem{cor}{Corollary}

\newcommand{\bx} {\boldsymbol{x}}

\newcommand{\by} {\boldsymbol{y}}

\newcommand{\bH} {\boldsymbol{H}}

\newcommand{\bI} {\boldsymbol{I}}
\newcommand{\bR} {\boldsymbol{R}}
\newcommand{\bU} {\boldsymbol{U}}
\newcommand{\bu} {\boldsymbol{u}}
\newcommand{\bh} {\boldsymbol{h}}

\newcommand{\bM} {\boldsymbol{M}}

\newcommand{\bLam} {\boldsymbol{\Lambda}}
\newcommand{\gl}{\lambda}

\newcommand{\bxi} {\boldsymbol{\xi}}

\def\bal#1\eal{\begin{align}#1\end{align}}
\newcommand{\bp} {\begin{proof}}
\newcommand{\ep} {\end{proof}}

\newcommand{{\bRF}} {\right\}}

\begin{document}

\title{The Capacity of Gaussian MIMO Channels Under Total and Per-Antenna Power Constraints}

\author{Sergey Loyka

\vspace*{-1\baselineskip}

\thanks{The material in this paper was presented in part at the IEEE International Symposium on Information Theory, Barcelona, Spain, July 2016.}

\thanks{S. Loyka is with the School of Electrical Engineering and Computer Science, University of Ottawa, Ontario, Canada, e-mail: sergey.loyka@ieee.org}

}

\maketitle

\begin{abstract}
The capacity of a fixed Gaussian multiple-input multiple-output (MIMO) channel and the optimal transmission strategy under the total power (TP) constraint and full channel state information are well-known. This problem remains open in the general case under individual per-antenna (PA) power constraints, while some special cases have been solved. These include a full-rank solution for the MIMO channel and a general solution for the multiple-input single-output (MISO) channel. In this paper, the fixed Gaussian MISO channel is considered and its capacity as well as optimal transmission strategies are determined in a closed form under the joint total and per-antenna power constraints in the general case. In particular, the optimal strategy is hybrid and includes two parts: first is equal-gain transmission and second is maximum-ratio transmission, which are responsible for the PA and TP constraints respectively. The optimal beamforming vector is given in a closed-form and an accurate yet simple approximation to the capacity is proposed. Finally, the above results are extended to the MIMO case by establishing the ergodic capacity of fading MIMO channels under the joint power constraints when the fading distribution is right unitary-invariant (of which i.i.d. and semi-correlated Rayleigh fading are special cases). Unlike the fixed MISO case, the optimal signaling is shown to be isotropic in this case.
\end{abstract}


%=====================================================================================
\section{Introduction}

The capacity of a fixed multiple-input multiple-output (MIMO) Gaussian channel under the total power (TP) constraint and full channel state information (CSI) at both ends is well-known as well as the optimal transmission strategy to achieve it \cite{Cover-06}-\cite{Biglieri-98}: the optimal strategy is Gaussian signaling over the channel eigenmodes with power allocation given by the water-filling (WF) algorithm. In the special case of multiple-input single-output (MISO) channel, this reduces to the rank-1 signalling, i.e. beamforming, where the beamforming vector is proportional to the channel vector (i.e. stronger channels get more power), which mimics the maximum ratio combining (MRC) in diversity reception systems \cite{Barry-04}\cite{Tse-05}, which we term here "maximum ratio transmission" (MRT). Recently, this problem was considered under individual per-antenna (PA) power constraints \cite{Yu-07}-\cite{Vu-11}, which is motivated by the distributed design of active antenna arrays where each antenna has its own RF amplifier with limited power (as opposed to a common amplifier and a passive beamforming network in the case of TP constraint\footnote{The following further considerations make the TP constraint important: (i) for battery-operated devices, the TP determines the battery life; (ii) the TP constraint is important when a power/energy supply is significantly limited; (iii) the growing importance of "green" communications makes the TP important since it is the TP rather than the PA power that determines the carbon footprint of the system.}), so that powers of different antennas cannot be traded off with each other. The optimal transmission strategy for a fixed channel was established in \cite{Vu-11}, which corresponds to beamforming (i.e. rank-1 transmission) with uniform amplitude distribution across antennas and where the beamforming vector compensates for channel phase differences so that all transmitted  signals are coherently combined at the receiver. This mimics the well-known equal gain combining (EGC) in a diversity-reception system. Hence, we term this strategy "equal gain transmission" (EGT) here. A fixed multiple-input multiple-output (MIMO) Gaussian channel under PA constraints was considered in \cite{Vu-11b} and \cite{Tuninetti-14}, where a numerical algorithm to evaluate an optimal Tx covariance was developed based on a partial analytical solution \cite{Vu-11b} and  a closed-form full-rank solution was obtained \cite{Tuninetti-14}, while the general solution remains illusive. This is in stark contrast to the capacity under the TP constraint, for which the general solution is well-known for this channel. The capacity of the ergodic-fading MISO channel under the long-term average PA constraint and full CSI at both ends was established in \cite{Maamari-14}.

Single-user PA-constrained results were extended to multi-user scenarios in \cite{Park-10} and \cite{Zhu-12}, where a precoder was developed that achieves a 2-user MISO Gaussian broadcast channel (BC) capacity \cite{Park-10} and an iterative numerical algorithm was developed to obtain optimal covariance matrices to maximize the sum-capacity of Gaussian MIMO multiple-access (MAC) channel \cite{Zhu-12}, for which no closed-form solution is known.

One may further consider a hybrid design of a Tx antenna array where each antenna has its own power amplifier and yet some power can be traded-off between antennas (corresponding to a common beamforming network) under the limited total power (e.g. due to the limitation of a power supply unit). This implies individual (PA) as well as total (TP) power constraints. Ergodic-fading MIMO channels were considered in \cite{Khoshnevisan-12} under long-term TP and short-term PA constraints and a sub-optimal signalling transmission strategy was proposed. An optimal strategy to achieve the ergodic capacity under the above constraints remains unknown. A fixed (non-fading) MISO channel was considered in \cite{Cao-15} under full CSI at both ends and joint TP and PA constrains. It was shown that beamforming is still an optimal strategy. A closed-form solution was established in the case of 2 Tx antennas only and the general case remains an open problem.

The present paper provides a closed-form solution to this open problem, which is based on Karush-Kuhn-Tucker (KKT) optimality conditions for the respective optimization problem. In particular, we show that the optimal strategy is hybrid and consists of 2 parts: 1st part, which includes antennas with stronger channel gains and for which PA constraints are active, performs EGT (when PA constraints are the same for all antennas) while 2nd part, which includes antennas with weaker channel gains and for which PA constraints are inactive, performs MRT. This mimics the classical equal gain and maximum ratio combining (EGC and MRC) strategies of diversity reception. Amplitude distribution across antennas as well as the number of active PA constraints are explicitly determined. Sufficient and necessary conditions for the optimality of the MRT and the EGT are given. In particular, the MRT is optimal when channel gain variation among antennas is not too large and the EGT is optimal for sufficiently large total power constraint.

Based on the fact that the capacity under the joint (PA+TP) constraints is upper bounded by the capacities under the individual (either PA or TP) constraints, a compact yet accurate approximation to the capacity is proposed.

While closed-form solutions for the optimal signaling and the capacity of the fixed Gaussian MISO channel under the joint power constraints are established in sections \ref{sec.C.PA.TP} and \ref{sec.diff.PA.const}, one may wonder whether they can be extended to the MIMO case and whether fading can be included as well, which is important from the practical perspective for modern wireless systems. Section \ref{sec.fading.MIMO} partially addresses this question by considering a class of fading MIMO channels and establishing its ergodic capacity  under the joint power constraints when the fading distribution is right unitary-invariant (see section V for details), of which i.i.d. and semi-correlated Rayleigh fading are special cases. Unlike the fixed MISO case, the optimal signaling is shown to be isotropic in this case. This extends the respective result in \cite{Telatar} established under the TP constraint and i.i.d. Rayleigh fading to the joint PA and TP constraints as well as to the class of right unitary-invariant fading distributions.

\textit{Notations}: bold lower-case letters denote column vectors, $\bh = [h_1,h_2,..,h_m]^T$, where $T$ is the transposition, while bold capital denote matrices; $\bR^+$ is the Hermitian conjugation of $\bR$; $r_{ii}$ denotes the $i$-th diagonal entry of $\bR$; $\lfloor x \rfloor$ is the integer part while $(x)_+=\max[0,x]$ is the positive part of $x$; $\nabla_R$ is the derivative with respect to $\bR$; $\bR \ge 0$ means that $\bR$ is positive semi-definite; $|\bh|_p = (\sum_i |h_i|^p)^{1/p}$ is the $l_p$-norm of vector $\bh$ and $|\bh|=|\bh|_2$ is the $l_2$ norm.

%======================================================================================
\section{Channel Model and Capacity}
\label{sec.Channel Model}

Discrete-time model of a fixed Gaussian MISO channel can be put into the following form:
\bal
%\label{eq.C.def}
y = \bh^+\bx +\xi
\eal
where $y, \bx, \xi$ and $\bh$ are the received and transmitted signals, noise and channel respectively; $h_i^*$ is $i$-th channel gain (between $i$-th Tx antenna and the Rx). Without loss of generality, we order the channel gains, unless indicated otherwise, as follows: $|h_1|\ge |h_2|\ge .. |h_m|>0$, and $m$ is the number of transmit antennas. The noise is assumed to be Gaussian with zero mean and unit variance, so that the SNR equals to the signal power. Complex-valued channel model is assumed throughout the paper, with full channel state information available both at the transmitter and the receiver. Gaussian signaling is known to be optimal in this setting \cite{Cover-06}-\cite{Biglieri-98} so that finding the channel capacity $C$ amounts to finding an optimal transmit covariance matrix $\bR$:
\bal
\label{eq.C.def}
C = \max_{\bR \in S_R} \ln(1+\bh^+\bR\bh)
\eal
where $S_R$ is the constraint set. In the case of the TP constraint, it takes the form
\bal
S_R=\{\bR: \bR\ge 0, tr\bR \le P_T\},
\eal
where $P_T$ is the maximum total Tx power, and the MRT is optimal \cite{Tse-05} so that the optimal covariance $\bR^*$ is
\bal
\bR^* = P_T \bh\bh^+/|\bh|_2^2
\eal
and the capacity is
\bal
\label{eq.C.MRT}
C_{MRT}= \ln(1+P_T |\bh|_2^2)
\eal
Under the PA constraints,
\bal
S_R=\{\bR: \bR\ge 0, r_{ii} \le P\},
\eal
where $r_{ii}$ is $i$-th diagonal entry of $\bR$ (the Tx power of $i$-th antenna), $P$ is the maximum PA power, and the EGT is optimal \cite{Vu-11} so that the optimal covariance $\bR^*$ is
\bal
\bR^* = P \bu\bu^+,
\eal
where the entries of the beamforming vector $\bu$ are $u_i=e^{j\phi_i}$, $\phi_i$ is the phase of $h_i$, and the capacity is
\bal
\label{eq.C.EGT}
C_{EGT}= \ln(1+P|\bh|_1^2)
\eal
Note from \eqref{eq.C.MRT} and \eqref{eq.C.EGT} that it is the $l_1$ norm of the channel $\bh$ that determines the capacity under the PA constraint while the $l_2$ norm does so under the total power constraint. In the next section, we will see how this observation extends to the case of the joint PA and TP constraints.

%=====================================================================================
\section{The Capacity Under the Joint Constraints}
\label{sec.C.PA.TP}

Following the same line of argument as for the total power constraint \cite{Cover-06}-\cite{Biglieri-98}, the channel capacity $C$ under the joint PA and TP constraints is as in \eqref{eq.C.def} where $S_R$ is as follows:
\bal
\label{eq.SR.PA.TP}
S_R=\{\bR: \bR\ge 0, tr\bR \le P_T, r_{ii}\le P\}
\eal
and $P_T, P$ are the maximum total and per-antenna powers. This is equivalent to maximizing the Rx SNR:
\bal
%\label{eq.C.def}
\max_{\bR} \bh^+\bR\bh\ \ \textup{s.t.}\ \ \bR \in S_R
\eal
The following Theorem gives a closed-form solution to this open problem.

%\newpage
\begin{thm}
\label{thm.C}
The MISO channel capacity in \eqref{eq.C.def} under the per-antenna and total power constraints in \eqref{eq.SR.PA.TP} is achieved by the beamforming with the following input covariance matrix
\bal
\label{eq.R*}
\bR^* = P^* \bu\bu^+
\eal
where $P^*=\min(P_T, mP)$ and $\bu$ is a unitary (beamforming) vector of the form:
\bal
\label{eq.ui*}
u_i &= a_i e^{j\phi_{i}}
\eal
where $\phi_i$ is the phase of $h_i$ and $a_i$ represents amplitude distribution across antennas:
\bal
a_i =
\begin{cases}
c_1,  &i=1..k\\
c_2 |h_i|, &i=k+1..m
\end{cases}
\eal
and
\bal
c_1 = \frac{1}{\sqrt{m^*}},\ c_2 = \frac{\sqrt{1-k/m^*}}{|\bh_{k+1}^m|_2}
\eal
$m^* = P^*/P$, $\bh_{k+1}^m = [h_{k+1}...h_m]^T$ is the truncated channel vector, and $k$ is the number of active per-antenna power constraints, $0\le k\le \lfloor m^*\rfloor$, determined as the least solution of the following inequality
\bal
\label{eq.hth}
|h_{k+1}|\le h_{th} = \frac{|\bh_{k+1}^m|_2}{\sqrt{m^*-k}}
\eal
if $P_T < mP$ and $k=m$ otherwise. The capacity is
\bal
\label{eq.C}
C = \ln(1+\gamma^*)
\eal
where $\gamma^* = \bh^+\bR^*\bh$ is the maximum Rx SNR under the TP and PA constraints,
\bal
\label{eq.gamma*}
\gamma^* = P^*(c_1|\bh_1^k|_1 + c_2|\bh_{k+1}^m|_2^2)^2
\eal
where the 2nd term is absent if $k=m$.
\end{thm}
\begin{proof}
see Appendix.
\end{proof}

Note from \eqref{eq.ui*} that the beamforming vector always compensates for channel phases so that the transmitted signals are combined coherently at the receiver, while the amplitude distribution across Tx antennas depends on the number of active PA constraints: amplitudes are always the same for those antennas for which PA constraints are active (which represent stronger channels) and they are proportional to channel gain when for inactive PA constraints (weaker channels). In accordance with this,
\eqref{eq.gamma*} has two terms: 1st term $c_1|\bh_1^k|_1$ represents the gain due to the equal gain transmission (EGT, $|u_i|=c_1$) for active PA constraints while 2nd one $c_2|\bh_{k+1}^m|_2^2$ - due to the maximum ratio transmission (MRT, $|u_i| = c_2|h_i|$) for inactive PA constraints, which mimic the equal gain combining (EGC) and maximum ratio combining (MRC) in the case of diversity reception systems. These two terms are represented by $l_1$ and $l_2$ norms respectively, which mimic the respective observation for \eqref{eq.C.EGT} and \eqref{eq.C.MRT}.

Eq. \eqref{eq.hth} facilitates an algorithmic solution to find the number $k$ of active PA constraints and hence the threshold $h_{th}$: the inequality is verified for $k$ in increasing order, starting from $k=0$, and the algorithm stops when 1st solution is found (this will automatically be the least solution, as required).

The following Corollary establishes conditions for the optimality of the MRT, which corresponds to $k=0$.

\begin{cor}
All PA constraints are inactive and thus maximum ratio transmission is the optimal strategy if and only if
\bal
\label{eq.cor.1}
|h_1| \le |\bh|_2 \sqrt{P/P_T}
\eal
\end{cor}
\begin{proof}
Follows directly from Theorem \ref{thm.C} by using $k=0$. The necessary part is due to the necessity of the KKT conditions for optimality.
\end{proof}

Note that this limits channel gain variance among antennas. In particular, it always holds if all channel gains are the same. It also implies that at least 1 PA constraint is active if
\bal
|h_1| > |\bh|_2 \sqrt{P/P_T}
\eal

In a similar way, one obtains a condition for the optimality of the EGT.

\begin{cor}
All PA constraints are active and thus the equal gain transmission is the optimal strategy if and only if
\bal
\label{eq.cor.2}
P_T \ge m P
\eal
\end{cor}
%\qed

When the TP constraint is not active, i.e. $P_T\ge mP$ and hence $k=m$, Theorem 1 reduces to the respective result in \cite{Vu-11} under the identical PA constraints.

\subsection{Examples}

To illustrate the optimal solution, we consider the following representative example: $\bh = [3,1,0.5,0.1]^T$. Note that this example also applies to complex-valued channel gains since the beamforming vector is always adjusted to compensate for the channel phases and hence they do not affect the capacity or the amplitude distribution, which will stay the same for the more general case of
\bal
\bh = [3e^{j\phi_1},e^{j\phi_2},0.5e^{j\phi_3},0.1e^{j\phi_4}]^T
\eal
where $\phi_1...\phi_4$ are (arbitrary) phases, which affect the beamforming vector phases as in \eqref{eq.ui*}. Fig. 1 shows the capacity under the total and joint power constraints as the function of the total power $P_T$ when $P=1$. As the total power increases, more and more PA constraints become active, starting with antennas corresponding to strongest channels. Note that the MRT is optimal ($k=0$) if the total power is not too large:
\bal
\label{eq.MRT.opt}
P_T \le P |\bh|^2/|h_1|^2 \approx 1.1
\eal
while the EGT is optimal if
\bal
\label{eq.EGT.opt}
P_T \ge m P =4
\eal

Fig. 2 shows the amplitude distribution for the scenario in Fig. 1 under the joint PA+TP constraints. While weak channels get less power at the beginning (when the MRT is optimal), it gradually increases as the strongest channels reach their individual power constrains until eventually all channels have the same power (when the EGT is optimal). Note that while the amplitudes $a_1$ and $a_4$ of the strongest and weakest channels are monotonically decreasing/increasing, the amplitudes $a_2, a_3$ of intermediate  channels are not monotonic in $P_T$, increasing first until they reach the stronger level and then decreasing.

\begin{figure}[t]%[htbp]
\label{Fig.1}
\centerline{\includegraphics[width=3in]{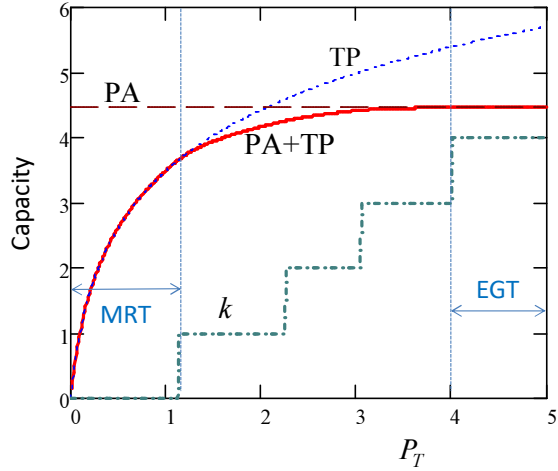}}
\caption{The capacity of MISO channel under the PA, TP and joint PA+TP constraints and the number of active PA constraints $k$ vs. total power $P_T$; $P=1$, $\bh = [3,1,0.5,0.1]^T$.}
\end{figure}

\begin{figure}[t]%[htbp]
\label{Fig.2}
\centerline{\includegraphics[width=3in]{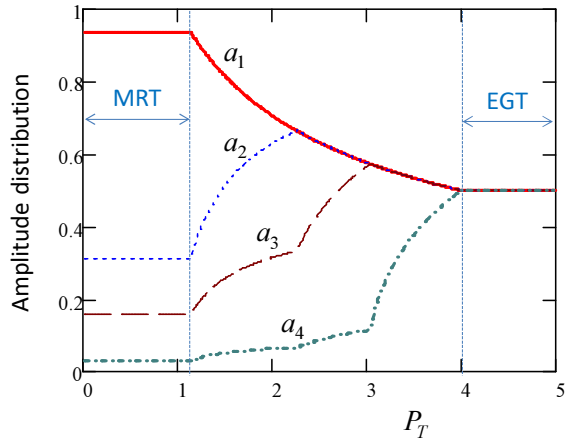}}
\caption{The optimal amplitude distribution under the joint power constraints for the scenario in Fig. 1.}
\end{figure}

In general, the capacity under the joint PA+TP constraints can be upper-bounded by the EGT and MRT capacities under the PA and TP constraints respectively:
\bal
\label{eq.C.ineq}
C \le \min(C_{MRT}, C_{EGT})
\eal
where $C_{MRT}, C_{EGT}$ are as in \eqref{eq.C.MRT}, \eqref{eq.C.EGT}, and the upper bound is tight everywhere except in the transition region, so one can approximate the capacity $C$ as
\bal
\label{eq.C.approx}
C \approx \min(C_{MRT}, C_{EGT})
\eal
It is straightforward to show that \eqref{eq.C.ineq} and \eqref{eq.C.approx} hold with strict equality under \eqref{eq.cor.1} or \eqref{eq.cor.2} for any $\bh$, or if $|h_1|/|h_m|=1$ for any $P_T$ and $P$. The approximation is sufficiently accurate if the variance in the channel gains is not large, i.e. if $|h_1|/|h_m|$ is not too large, as the following example demonstrates in Fig. 3, where $\bh = [4,3,2.5,2]^T$. Fig. 4 shows the respective amplitude distribution. Notice that the variance of the amplitude distribution is smaller than that in Fig. 2, since the variance in the channel gains is smaller as well, and that the range of optimality of the MRT is larger while the range of optimality of the EGT is exactly the same as in Fig. 2. In fact, it follows from \eqref{eq.MRT.opt} and \eqref{eq.EGT.opt} that while the range of optimality of the MRT depends on the channel, that of the EGT does not.

\begin{figure}[t]%[htbp]
\label{Fig.3}
\centerline{\includegraphics[width=3in]{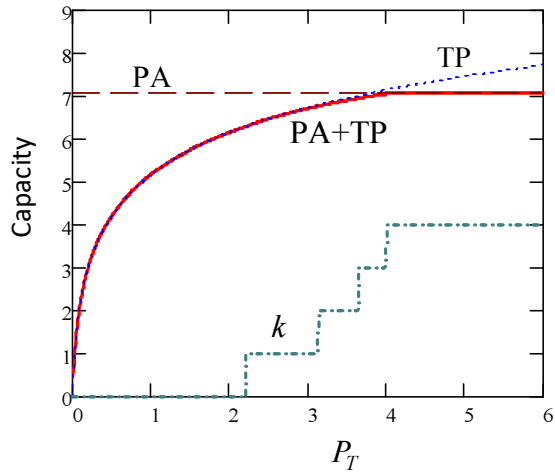}}
\caption{The capacity of MISO channel under the PA, TP and joint PA+TP constraints and the number of active PA constraints $k$ vs. total power $P_T$; $P=1$, $\bh = [4,3,2.5,2]^T$. Note that the approximation in \eqref{eq.C.approx} is accurate over the whole range of $P_T$.}
\end{figure}

\begin{figure}[t]%[htbp]
\label{Fig.4}
\centerline{\includegraphics[width=3in]{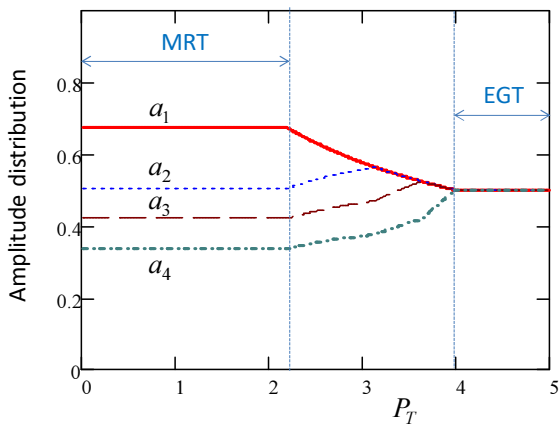}}
\caption{The optimal amplitude distribution under the joint power constraints for the scenario in Fig. 3.}
\end{figure}

%\newpage
%======================================================================================
\section{Different PA Constraints}
\label{sec.diff.PA.const}

In a similar way, one may wish to consider a more general case where individual antennas have different power constraints, so that the constraint set is
\bal
\label{eq.SR.Pi}
S_R=\{\bR: \bR\ge 0, tr\bR \le P_T, r_{ii}\le P_i\}
\eal
The channel capacity under these constraints is given in the following.

\begin{thm}
\label{thm.C.Pi}
The MISO channel capacity in \eqref{eq.C.def} under the per-antenna and total power constraints in \eqref{eq.SR.Pi} is achieved by the beamforming with the input covariance matrix as in \eqref{eq.R*} and \eqref{eq.ui*} where
\bal
a_i =
\begin{cases}
c_{1i},  &i=1..k\\
c_2 |h_i|, &i=k+1..m
\end{cases}
\eal
\bal
c_{1i} = \sqrt{\frac{P_i}{P^*}},\ c_2 = \frac{\sqrt{1-k/m^*}}{|\bh_{k+1}^m|_2}
\eal
and $P^*=\min(P_T, \sum_{i=1}^m P_i)$, $m^* = P^*/P_0$, $P_0=\frac{1}{k}\sum_{i=1}^k P_i$ is the average power of the active PA constraints, $k$ is the number of active PA constraints, determined as the least solution of the following inequality
\bal
\frac{|h_{k+1}|}{\sqrt{P_{k+1}}} \le \frac{|\bh_{k+1}^m|_2}{\sqrt{P_T-\sum_{i=1}^k P_i}}
\eal
if $P_T < \sum_{i=1}^m P_i$ and $k=m$ otherwise, where channel gains $\{h_i\}$ are ordered in such a way that $\{|h_{i}|/\sqrt{P_{i}}\}$ are in decreasing order. The capacity is as in \eqref{eq.C} and the optimal SNR is
\bal
\label{eq.gamma*2}
\gamma^* = P^*\left(\sum_{i=1}^k c_{1i} |h_i| + c_2|\bh_{k+1}^m|_2^2\right)^2
\eal
\end{thm}
\begin{proof}
Follows along the same lines as that of Theorem \ref{thm.C}.
\end{proof}

Note that 1st term in \eqref{eq.gamma*2} does not represent EGT anymore; rather, the amplitudes are adjusted to match the PA constraints. The conditions for optimality of the MRT can be similarly obtained. When the TP constraint is inactive, i.e. when $P_T \ge \sum_{i=1}^m P_i$, Theorem 2 reduces to the respective result in \cite{Vu-11}, as it should be. The condition for the optimality of the MRT is as follows.

\begin{cor}
All PA constraints are inactive and thus the MRT is optimal if and only if
\bal
|h_1| \le |\bh|_2 \sqrt{P_1/P_T}
\eal
and at least 1 PA constraint is active otherwise. All PA constraints are active if and only if
\bal
P_T \ge \sum_{i=1}^m P_i
\eal
\end{cor}
%\qed

%=====================================================================================
\section{Fading MIMO Channels}
\label{sec.fading.MIMO}

While the closed form solutions for the optimal signaling and the capacity of fixed MISO channels under the joint power constraints have been obtained above, one may wonder whether they can be extended to the MIMO case and whether fading can be included as well, which is of particular importance for modern wireless systems.

In this section, we partially answer this question by considering Gaussian fading MIMO channels of the form
\bal
\label{eq.MIMO.ch}
\by = \bH \bx + \bxi
\eal
where $\bx, \by$ are the transmitted and received (vector) signals, $\bxi$ is the Gaussian i.i.d. noise and $\bH$ is the channel matrix. The entries of this matrix are random variables representing fading channel gains between each transmit and each receive antenna. We assume that the Tx has the channel distribution information only (due to e.g. limitations of the feedback link and the channel estimation mechanism, see e.g. \cite{Biglieri}). A class of ergodic fading distributions will be considered, of which i.i.d. Rayleigh fading is a special case. The following definition characterizes this class.

\begin{defn}
A fading distribution of $\bH$ is right unitary-invariant if $\bH\bU$ and $\bH$ are equal in distribution for any unitary matrix $\bU$ of appropriate size.
\end{defn}

To see a physical motivation behind this definition, observe that i.i.d. Rayleigh fading, where each entry of $\bH$ is i.i.d. complex Gaussian with zero mean, satisfies this condition. A more general class of distributions which fit into this definition can be obtained by considering the popular Kronecker correlation model, see e.g. \cite{Kermoal}, where the overall channel correlation is a product of the independent Tx and Rx parts, which are induced by the respective sets of scatterers (e.g. around the base station and mobile unit), so that the channel matrix is
\bal
\bH = \bR_r^{1/2}\bH_0\bR_t^{1/2}
\eal
where $\bR_r, \bR_t$ are the Rx and Tx end correlations and the entries of $\bH_0$ are i.i.d. complex Gaussian with zero mean. While this model does not fit in general into Definition 1, its special case of no Tx correlation, $\bR_t=\bI$, so that
\bal
\label{eq.MIMO.H0}
\bH = \bR_r^{1/2}\bH_0
\eal
is indeed right unitary-invariant (since $\bH_0$ and $\bH_0\bU$ have the same distribution). Note that this model does allow an (arbitrary) Rx correlation. The uncorrelated Tx end may represent a base station where the antennas are spaced sufficiently widely apart of each other thereby inducing independence, see e.g. \cite{Loyka-02}.

The following Theorem establishes the ergodic capacity of a Gaussian MIMO channel under a right unitary-invariant fading distribution and the joint PA and TP constraints.

\begin{thm}
\label{thm.MIMO.C}
Consider the ergodic-fading MIMO channel as in \eqref{eq.MIMO.ch} for which the fading distribution is right unitary-invariant. Its channel capacity under the joint PA and TP constraints in \eqref{eq.SR.PA.TP} is as follows:
\bal
C = \mathbb{E}_{\bH}\{\ln|\bI + P^* \bH\bH^+|\}
\eal
where $\mathbb{E}_{\bH}$ is the expectation with respect to the fading distribution, $P^*=\min\{P,P_T/m\}$, and the optimal Tx covariance matrix is $\bR^*=P^*\bI$, i.e. isotropic (independent) signaling is optimal.
\end{thm}
\begin{proof}
The proof consists of two parts. In Part 1, we establish the optimality of isotropic signaling under the TP constraint only, while in Part 2, we extend this result to include the PA constraints as well.

\textit{Part 1}: the ergodic capacity under the TP constraint can be presented in the following form:
\bal
\label{eq.MIMO.P.0}
C_1 &= \max_{\bR}\mathbb{E}_{\bH}\{\ln|\bI + \bH\bR\bH^+|\}\\
\label{eq.MIMO.P.1}
&= \max_{\bR}\mathbb{E}_{\bH}\{\ln|\bI + \bH\bU\bLam\bU^+\bH^+|\}\\
\label{eq.MIMO.P.2}
&= \max_{\bR}\mathbb{E}_{\widetilde{\bH}}\{\ln|\bI + \widetilde{\bH}\bLam\widetilde{\bH}^+|\}\\
\label{eq.MIMO.P.3}
&= \max_{\bLam}\mathbb{E}_{\bH}\{\ln|\bI + \bH\bLam\bH^+|\}
\eal
where the maximization is subject to $\bR \ge 0$, $tr\bR \le P_T$. \eqref{eq.MIMO.P.0} is the standard expression for the ergodic MIMO channel capacity, see e.g. \cite{Telatar}\cite{Biglieri}; \eqref{eq.MIMO.P.1} follows from the eigenvalue decomposition $\bR=\bU\bLam\bU^+$, where the columns of unitary matrix $\bU$ are the eigenvectors of $\bR$ and the diagonal matrix $\bLam$ collects the eigenvalues of $\bR$; \eqref{eq.MIMO.P.2} follows from $\widetilde{\bH} = \bH\bU$; \eqref{eq.MIMO.P.3} follows since $\widetilde{\bH}$ and $\bH$ have the same distribution and the constraint $tr\bR =tr\bLam \le P_T$ depends only on the eigenvalues and hence the eigenvectors can be eliminated from the optimization. To proceed further, let
\bal
C(\bLam)= \mathbb{E}_{\bH}\{\ln|\bI + \bH\bLam\bH^+|\}
\eal
and observe that this is a concave function (since $\ln|\cdot|$ is and $\mathbb{E}_{\bH}$ preserves concavity, see e.g. \cite{Boyd-04}). Further observe the following chain inequality:
\bal
\label{eq.MIMO.P.4}
C(\bLam)&= \mathbb{E}_{\bH}\{\ln|\bI + \bH\bLam_{\pi}\bH^+|\}\\
\label{eq.MIMO.P.5}
&= \frac{1}{m!}\sum_{\pi} \mathbb{E}_{\bH}\{\ln|\bI + \bH\bLam_{\pi}\bH^+|\}\\
\label{eq.MIMO.P.6}
&\le \mathbb{E}_{\bH}\left\{\ln\left|\bI + \frac{1}{m!}\sum_{\pi} \bH\bLam_{\pi}\bH^+\right|\right\}\\
\label{eq.MIMO.P.7}
&= \mathbb{E}_{\bH}\{\ln|\bI + P^*\bH\bH^+|\}
\eal
where $\bLam_{\pi}$ is a diagonal matrix whose diagonal entries are a permutation $\pi$ of those in $\bLam$, and $P^*=P_T/m$. \eqref{eq.MIMO.P.4} follows from the fact that a permutation can be represented by a unitary matrix (where each column and each row has all zero entries except for one) and hence $C(\bLam)=C(\bLam_{\pi})$; \eqref{eq.MIMO.P.5} follows since \eqref{eq.MIMO.P.4} holds for any $\pi$ and the total number of permutations is $m!$; the inequality in \eqref{eq.MIMO.P.6} is due to the concavity of $C(\bLam)$; \eqref{eq.MIMO.P.7} follows from $\frac{1}{m!}\sum_{\pi}\bLam_{\pi} = P^*\bI$. Since the inequality in \eqref{eq.MIMO.P.6} becomes equality when $\bLam=P^*\bI$, the optimal signaling under the TP constraint is $\bR^*=P^*\bI$ so that
\bal
C_1= \max_{\bLam} C(\bLam) = \mathbb{E}_{\bH}\{\ln|\bI + P^*\bH\bH^+|\},
\eal
which establishes Part 1.

\textit{Part 2:} consider first the case when $P\ge P_T/m$ and observe that the capacity under the joint constraints $C_2$ cannot exceed that under the TP constraint only, which hence serves as an upper bound: $C_2\le C_1$. Since the TP optimal covariance $\bR^*=P_T\bI/m$ also satisfies the PA constraints (under the assumed condition), it is also optimal under the joint constraints and hence the upper bound is achieved: $C_2=C_1$. If, on the other hand, $P< P_T/m$, observe that the TP constraint is redundant (since, due to the PA constraints, the total power does not exceed $mP<P_T$) and hence the jointly-constrained optimization with $P_T>mP$ is equivalent to the PA-constrained optimization only, which in turn is equivalent to the jointly-constrained optimization with new total power $P_T'=mP$ (since the new TP constraint is also redundant). However, the latter problem is just a special case of $P\ge P_T/m$ considered above, from which the optimality of $\bR^*=P_T'\bI/m=P\bI$ follows.

Combining two parts, it follows that $\bR^*=\min\{P,P_T/m\}\bI$ is optimal under the joint constraints in general and hence $C_1=C_2$. This completes the proof.
\end{proof}

It follows from Theorem \ref{thm.MIMO.C} and its proof that the same capacity expression holds under the TP constraint, the PA constraints and the joint PA and TP constraints (where $P^*$ is defined accordingly). This extends the earlier result in \cite{Telatar} established for i.i.d. Rayleigh fading  and the TP constraint to the class of right unitary-invariant fading distributions (including, as a special case, the semi-correlated model in \eqref{eq.MIMO.H0}) and to the PA as well as the joint PA and TP constraints. Note that the optimal signaling here is isotropic, so that the optimal covariance matrix is full-rank, unlike that in Theorem \ref{thm.C}, which is of rank-1. The importance of isotropic signalling is due to the fact that no channel state or distribution information is needed at the Tx end (and hence the feedback requirements are minimal).

Applying this Theorem to the semi-correlated channel in \eqref{eq.MIMO.H0}, one obtains its ergodic capacity under the joint PA and TP constraints:
\bal
\label{eq.MIMO.C0}
C = \mathbb{E}_{\bH_0}\{\ln|\bI + P^* \bR_r\bH_0\bH_0^+|\}
\eal
Unlike the fixed MISO case, the optimal signaling here is isotropic, $\bR^*=P^*\bI$,  and hence independent of the Rx correlation $\bR_r$. However, the capacity does depend on $\bR_r$ but, as it follows from \eqref{eq.MIMO.C0}, $C$ depends on the eigenvalues of $\bR_r$ only, not on its eigenvectors:
\bal\notag
%\label{eq.MIMO.P.0}
C &= \mathbb{E}_{\bH_0}\{\ln|\bI + P^*\bU_r\bLam_r\bU_r^+\bH_0\bH_0^+|\}\\
&= \mathbb{E}_{\tilde{\bH}_0}\{\ln|\bI + P^*\bLam_r\tilde{\bH}_0\tilde{\bH}_0^+|\}\\ \notag
&= \mathbb{E}_{\bH_0}\{\ln|\bI + P^*\bLam_r\bH_0\bH_0^+|\}
\eal
where $\bR_r=\bU_r\bLam_r\bU_r^+$ is the eigenvalue decomposition, and $\tilde{\bH}_0 = \bU_r^+\bH_0$. The last equality is due to the fact that $\tilde{\bH}_0$ and $\bH_0$ are equal in distribution. Hence, different $\bR_r$ induce the same capacity provided that they have the same eigenvalues.

These properties are ultimately due to the right unitary invariance of the fading process. It can be further shown (by examples) that Theorem \ref{thm.MIMO.C} does not hold in general if fading distribution is not right unitary-invariant: e.g. consider $\bR_t=diag\{1,0,...,0\}$ for which the optimal covariance can be shown to be $\bR^*=diag\{\min(P_T, P),0,...,0\}$ (i.e. all the power is allocated to the only non-zero eigenmode of $\bR_t$).

%===================================================================================
\section{Conclusion}

The Gaussian MISO channel has been considered under the joint total and per-antenna power constraints. Its capacity as well as the optimal transmission strategy have been established in closed-form, thus extending earlier results established under individual constraints only or, in the case of joint constraints, for 2 Tx antennas only. It is interesting to observe that the optimal transmission strategy is hybrid, i.e. a combination of equal gain (for stronger antennas) and maximum-ratio (for weaker antennas) transmission strategies. If the variance of channel gains across antennas is not too large, the maximum ratio transmission is optimal and individual power constraints are not active. Finally, the above results have been extended to the MIMO case by establishing the ergodic capacity of fading MIMO channels under the joint power constraints when the fading distribution is right unitary invariant, which includes, as special cases, i.i.d. and semi-correlated Rayleigh fading. The optimal signaling in this case has been shown to be isotropic and hence the feedback requirements are minimal.

%===================================================================================
\section{Acknowledgement}

The author is grateful to R.F. Schaefer for insightful discussions, and to V.I. Mordachev and T. Griken for their support.

%===================================================================================
\section{Appendix}

\subsection{Proof of Theorem \ref{thm.C}}

The problem in \eqref{eq.C.def} under the constraints in \eqref{eq.SR.PA.TP} is convex (since the objective is affine and the constraints are affine and positive semi-definite). Since Slater's condition holds, KKT conditions are sufficient for optimality \cite{Boyd-04}. The Lagrangian for this problem is:
\bal
L = -\bh^+\bR\bh +\gl(tr\bR-P_T) + \sum_i\gl_i(r_{ii}-P) - tr\bM\bR
\eal
where $\gl, \gl_i \ge 0$ are Lagrange multipliers responsible for the total and per-antenna power constraints, and $\bM\ge 0$ is (matrix) Lagrange multiplier responsible for the positive semi-definite constraint $\bR\ge 0$. The KKT conditions are
\bal
\label{eq.KKT.1}
\nabla_R L = -\bh\bh^+ +\gl\bI -\bM +\bLam=0\\
\label{eq.KKT.2}
\gl(tr\bR-P_T)=0,\ \gl_i(r_{ii}-P)=0,\ \bR\bM=0\\
\label{eq.KKT.3}
tr\bR \le P_T,\ r_{ii}\le P,\\
\label{eq.KKT.4}
\bM\ge 0,\ \gl_i\ge 0
\eal
where $\nabla_R$ is the derivative with respect to $\bR$ and $\bLam=diag\{\gl_1...\gl_m\}$ is a diagonal matrix collecting $\gl_i$; \eqref{eq.KKT.1} is the stationarity condition, \eqref{eq.KKT.2} are complementary slackness conditions; \eqref{eq.KKT.3} and \eqref{eq.KKT.4} are primal and dual feasibility conditions.

Combining both inequalities in \eqref{eq.KKT.3}, one obtains:
\bal
tr\bR \le \min(P_T,mP)=P^*
\eal
and from \eqref{eq.KKT.1}
\bal
%\label{eq.KKT.1}
\bh\bh^+ +\bM = \gl\bI +\bLam > 0
\eal
where the last inequality is due to the diagonal part of the equality:
\bal
\label{eq.gl+gli>0}
|h_i|^2+m_{ii} = \gl +\gl_i >0
\eal
since $m_{ii} \ge 0$ and $|h_i|>0$. Therefore, $\bh\bh^+ +\bM$ is full-rank, $r(\bh\bh^+ +\bM)=m$. Since $r(\bh\bh^+)=1$ and $\bM\ge 0$, it follows that $r(\bM)\ge m-1$. Since $r(\bM)=m$ implies $\bM>0$ and hence $\bR=0$ - clearly not an optimal solution, one concludes that $r(\bM)=m-1$ and hence $r(\bR)=1$ (this follows from complementary slackness $\bM\bR=0$), i.e. beamforming is optimal:
\bal
%\label{eq.R*}
\bR^*=P^* \bu\bu^+
\eal
where $|\bu|=1$. It remains to find the beamforming vector $\bu$. To this end, combining the last equation with $\bM\bR=0$, one obtains:
\bal
0 = \bM\bu = -\bh^+\bu \bh + (\bLam+\gl\bI)\bu
\eal
from which it follows that
\bal
u_i = \bh^+\bu h_i/(\gl+\gl_i)
\eal
and hence
\bal
\phi_{ui} = \phi_i + \varphi = \phi_i
\eal
where $\phi_{ui}, \varphi$ are the phases of $u_i$ and $\bh^+\bu$; since the common phase $\varphi$ does not affect $\bR$ or the SNR, one can set $\varphi=0$ without loss of generality to obtain
\bal
\label{eq.ui.1}
u_i = a h_i/(\gl+\gl_i)
\eal
where $a=|\bh^+\bu|$.

If $\gl_i>0$ (active $i$-th per-antenna constraint), then $r_{ii}=P^*|u_i|^2=P$ from \eqref{eq.KKT.2} and \eqref{eq.R*} so that
\bal
|u_i|=c_1=1/\sqrt{m^*}
\eal
Since $\gl_i>0$, using \eqref{eq.ui.1},
\bal
c_1= |u_i| = a |h_i|/(\gl+\gl_i) < a |h_i|/\gl
\eal
so that
\bal
\label{eq.h1-k}
|h_i| > \gl c_1/a = h_{th}
\eal
where $h_{th}$ is a threshold channel gain, i.e. PA constraints are active for all sufficiently strong channels.

When $\gl_i=0$ (inactive $i$-th PA constraint) for at least one $i$, it follows from \eqref{eq.gl+gli>0} that $\gl>0$, i.e. the TP constraint is active: $tr\bR=P_T$, which implies $P_T \le mP$. One obtains from \eqref{eq.ui.1} in this case
\bal
\label{eq.ui.2}
u_i=c_2 h_i,\ c_2=a/\gl
\eal
which, when combined with the PA constraint $r_{ii} = P_T|u_i|^2 \le P$, requires
\bal
\label{eq.h-k+1}
|h_i| \le h_{th}
\eal
where $c_2$ can be found from the TP constraint $|\bu|^2=1$:
\bal
|\bu|^2 = k c_1^2 + c_2^2|\bh_{k+1}^m|^2=1
\eal
and $k<m$ is the number of active PA constraints, i.e. when \eqref{eq.h1-k} holds, which implies
\bal
c_2 = \frac{\sqrt{1-k/m^*}}{|\bh_{k+1}^m|_2}
\eal
so that $k\le m^*$ and $h_{th}$ can be expressed as
\bal
\label{eq.hth.2}
h_{th} = \frac{\gl c_1}{a} = \frac{c_1}{c_2} = \frac{|\bh_{k+1}^m|_2}{\sqrt{m^*-k}}
\eal
If $k=m$, i.e. all PA constraints are active, then one can take $h_{th}=0$ for consistency with \eqref{eq.h1-k}. This implies $P_T\ge mP$ so that $m^*=m$ (note that \eqref{eq.hth.2} is not defined in this case).

To find the number $k$ of active PA constraints when $P_T < mP$, so that $m^*=P_T/P < m$ and hence $k\le m^* <m$, observe that \eqref{eq.h1-k} and \eqref{eq.ui.2} imply
\bal
\label{eq.h1-k.2}
|h_{k}|\sqrt{m^*-k} >  |\bh_{k+1}^m|_2
\eal
while \eqref{eq.h-k+1} implies
\bal
\label{eq.h-k+1.2}
|h_{k+1}|\sqrt{m^*-k} \le |\bh_{k+1}^m|_2
\eal
both due to the ordering $|h_1|\ge|h_2|\ge ..\ge |h_m|$, so that $k$ has to satisfy both inequalities simultaneously.

The next step is to show that there exists unique $k$ that satisfies both inequalities. First, we show that there is at least one solution of \eqref{eq.h-k+1.2}.

\begin{lemma}
\label{lem.k.1}
There exists at least one solution $k$, $0\le k\le m^*$, of \eqref{eq.h-k+1.2}.
\end{lemma}
\begin{proof}
If $m^*= m$, then $k=m$ clearly solves it, where we take $h_{m+1}=0$ for consistency (recall that all channels with 0 gain do not affect the capacity). If $m^* < m$, then $k= \lfloor m^* \rfloor$ solves it.
\end{proof}

The next Lemma shows that, in general, a solution is not unique.

\begin{lemma}
If $k\le \lfloor m^* \rfloor$ satisfies \eqref{eq.h-k+1.2}, then all $k'$ such that $k \le k' \le \lfloor m^* \rfloor$ also satisfy it, i.e. a solution is not unique in general. Likewise, all $k' \le k$ solve \eqref{eq.h1-k.2} if $k$ solves it.
\end{lemma}
\begin{proof}
Let \eqref{eq.h-k+1.2} to hold for $k<\lfloor m^* \rfloor$, so that
\bal
%\label{eq.h-k+1.2}
|h_{k+1}|^2(m^*-k) \le  |h_{k+1}|^2+..+|h_{m}|^2
\eal
and hence
\bal\notag
%\label{eq.h-k+1.2}
|h_{k+2}|^2(m^*-(k+1))&\le |h_{k+1}|^2(m^*-(k+1))\\
 &\le  |h_{k+2}|^2+..+|h_{m}|^2
\eal
i.e. \eqref{eq.h-k+1.2} also holds for $k'=k+1$. By induction, it holds for all $k\le k' \le \lfloor m^* \rfloor$. To prove 2nd claim, note that it follows from \eqref{eq.h1-k.2} that
\bal\notag
%\label{eq.h1-k.2}
|h_{k-1}|^2(m^*-k+1)&\ge |h_{k}|^2(m^*-k+1) \\
&>  |h_{k}|^2+..+|h_m|^2
\eal
\end{proof}

Finally, we show that a unique $k$ satisfying both inequalities does exist.

\begin{prop}
There exists a unique solution of \eqref{eq.h1-k.2} and \eqref{eq.h-k+1.2}, which is also the least solution of \eqref{eq.h-k+1.2}.
\end{prop}
\begin{proof}
Note, from Lemma \ref{lem.k.1}, that a least solution $k'$ of \eqref{eq.h-k+1.2} exists, so that the following holds
\bal
%\label{eq.h1-k.2}
|h_{k'+1}|^2(m^*-k')\le |h_{k'+1}|^2+..+|h_m|^2\\
|h_{k'}|^2(m^*-k+1)> |h_{k'}|^2+..+|h_m|^2
\eal
where the last inequality is due to the fact that $k'$ is the least solution; this inequality implies
\bal
%\label{eq.h1-k.2}
|h_{k'}|^2(m^*-k)> |h_{k'+1}|^2+..+|h_m|^2
\eal
i.e. \eqref{eq.h1-k.2} holds for $k=k'$.
\end{proof}

It remains to show that $\bM \ge 0$ (dual feasibility). To this end, note that this is equivalent to $\bx^+\bM\bx \ge 0\ \forall \bx$. It follows from \eqref{eq.KKT.1}, \eqref{eq.ui.1}, \eqref{eq.h1-k}, \eqref{eq.ui.2} and Caushy-Schwarz inequality that
\bal
%\label{eq.KKT.1}
\bx^+\bM\bx &= -|\bh^+\bx|^2 +\gl|\bx|^2 + \bx^+\bLam\bx\\
&\ge -|\bh^+\bx|^2 + \frac{a}{c_1} \sum_{i=1}^m |h_i||x_i|^2 \\
&\ge -|\bh^+\bx|^2 + |h_1| \sum_{i=1}^m |h_i||x_i|^2 \\
&\ge -|\bh^+\bx|^2 + \sum_{i=1}^m |h_i|^2|x_i|^2 \ge 0
\eal
This completes the proof.

\end{document}